\documentclass{article}

\usepackage{graphicx}
\usepackage{amsfonts}
\usepackage{amssymb}
\usepackage{amsmath}
\usepackage{verbatim}
\usepackage{float}
\usepackage{scrextend}
\usepackage{xspace}
\usepackage{fullpage}

\newtheorem{definition}{Definition}
\newtheorem{theorem}{Theorem}
\newtheorem{lemma}{Lemma}
\newtheorem{specification}{Specification}

\newenvironment{proof}{{\bf Proof. } }{{\hfill $\Box$}}
\deffootnote[1em]{1em}{0em}{\textsuperscript{\thefootnotemark}\,}
\floatstyle{ruled}
\newfloat{algorithm}{thp}{lop}
\floatname{algorithm}{Algorithm}

\newcommand{\ie}{{\em i.e.,}\xspace}
\newcommand{\eg}{{\em e.g.,}\xspace}

\begin{document}

\title{
Enabling Minimal Dominating Set in\\Highly Dynamic Distributed Systems
}

\renewcommand*{\thefootnote}{\fnsymbol{footnote}}

\author{
Swan Dubois\footnotemark[1]
\and
Mohamed-Hamza Kaaouachi\footnotemark[1]
\and
Franck Petit\footnotemark[1]
}

\footnotetext[1]{Sorbonne Universit\'es, UPMC Universit\'e Paris 6, F-75005, Paris, France\\
CNRS, UMR 7606, LIP6, F-75005, Paris, France\\
Inria, \'Equipe-projet REGAL, F-75005, Paris, France\\
E-mail: {\tt firstname.lastname@lip6.fr}}

\renewcommand*{\thefootnote}{\arabic{footnote}}
\setcounter{footnote}{0}

\date{}

\maketitle

\begin{abstract}
We address the problem of computing a Minimal Dominating Set in highly dynamic distributed systems. We assume weak connectivity, \ie the network may be disconnected at each time instant and topological changes are unpredictable. We make only weak assumptions on the communication: every process is infinitely often able to communicate with other processes (not necessarily directly). 

Our contribution is threefold. First, we propose a new definition of minimal dominating set suitable for the context of time-varying graphs that seems more relevant than existing ones. Next, we provide a necessary and sufficient topological condition for the existence of a deterministic algorithm for minimal dominating set construction in our settings. Finally, we propose a new measure of time complexity in time-varying graph in order to  to allow fair comparison between algorithms. Indeed, this measure takes account of communication delays attributable to dynamicity of the graph and not to the algorithms. 
\end{abstract}

\section{Introduction} 

The availability of wireless communications has drastically increased in recent years and established new applications. Humans, agents, devices, robots, and applications interact together through more and more heterogeneous infrastructures, such as mobile \emph{ad hoc} networks (MANET), vehicular networks (VANET), (mobile) sensor and actuator networks (SAN), body area networks (BAN), as well as always evolving network infrastructures on the Internet. In such networks, items (users, links, equipments, {\em etc.}) may join, leave, or move inside the network at unforeseeable times. A common feature of these networks is their {\em high dynamic}, meaning that their topology keeps continuously changing over time. Dynamic, heterogeneity of devices, usages, and participants, and often the unprecedented scale to consider, make the design of such infrastructures extremely challenging. For a vast majority of them, the dynamics are also unpredictable. Classically, distributed systems are modeled by a static undirected connected graph where vertices are processes (nodes, servers, processors, etc.) and edges represent bidirectional communication links.  Clearly, such modeling is not suitable for high dynamic networks.

Numerous models taking in account topological changes over time have have been proposed since several decades, \eg \cite{AKMUV12,AE84,CCF09,F03,F04,FGM07,SW09}. Some works aim at unifying most of the above approaches. For instance, in~\cite{XFJ03}, the authors introduced the {\em evolving graphs}. They proposed modeling the time as a sequence of discrete time instants and the system dynamic by a sequence of static graphs, one for each time instant. More recently, another graph formalism, called {\em Time-Varying Graphs} (TVG), has been provided in~\cite{CFQS12}. In contrast with evolving graphs, TVGs allow systems evolving within continuous time. Also in~\cite{CFQS12} and in companion papers~\cite{CFMS10,CFMS12}, TVGs are gathered and ordered into classes depending mainly on two main features: the quality of connectivity among the participating nodes and the possibility/impossibility to perform tasks.

In this paper, we focus on the {\em Minimal Dominating Set} (MDS) problem. A dominating set is a subset of vertices of a graph such as each vertex of this graph is either in the dominating set or neighbor of a vertex in the dominating set. A \emph{minimal} dominating set is such that none of its proper subsets is also a dominating set of the graph. Like many distributed covering structure (such as trees, coloring, matching, \emph{etc.}), Minimal Dominating Set is a key building block for numerous network protocols, \eg hierarchical routing and clustering, unicast, multicast, topology control, media access coordination, to name only a few.

Minimal Dominating Set and some of related problems (such as Maximal Independent Set and Connected Dominating Set) receive some attention in the context of dynamic networks, \eg \cite{BDTC05,WDCG12,SW10,CMM11}. The difficulty to define covering structures in dynamic networks (including MDS) is pointed out in~\cite{CF13r}. Indeed, the authors show that the definition of such structures may become ambiguous, incorrect, or even irrelevant when applied in dynamic systems. As an example, if the dynamicity of the graph is modeled as a sequence of static graphs and a new MDS is computed at each topological change as in \cite{WDCG12}, the stability of the MDS fully depends on the dynamic rate of the network (\ie the relative speed of appearance/disappearance of edges). This natural definition may hence lead to an high instability (or even impossibility of use) of the MDS. We discuss more precisely this issue in Section~\ref{sec:MDS}.

This paper aims at proposing a new approach suitable for Minimal Dominating Set construction in time-varying graph with weak connectivity, \ie the graph may be disconnected at each time instant and topological changes are unpredictable. The only assumption on communications is that every process is infinitely often able to communicate with other processes (not necessarily directly). In this context, our contribution is threefold. First, we propose a new definition of MDS for time-varying graphs that increases stability of this structure. More precisely, we require that each dominated node is infinitely often neighbor of at least one dominating node. Next, we provide a necessary and sufficient topological condition for the existence of a deterministic algorithm for MDS construction in our settings. Finally, we propose a new measure of time complexity in time-varying graph. This measure takes account of communication delays attributable to the dynamicity of the graph and not to the algorithm in order to allow fair comparison between algorithms.

The paper is organized as follows. Section \ref{sec:TVG} presents formally the time-varying graph model and our new measure of time complexity. We devote the Section \ref{sec:UG} to some preliminaries necessary to our main results on MDS presented in Section \ref{sec:MDS}. Finally, Section \ref{sec:conclu} concludes the paper.

\section{Time-Varying Graph: Model and Complexity}\label{sec:TVG}

This section aims to present formally the framework of our study of dynamic systems. In a first time, we recall in Section \ref{sub:model} the model of time-varying graphs (TVGs) introduced by \cite{CFQS12}. We present only definitions needed for the comprehension of our work and we refer the reader to \cite{CFQS12} for more details and an interesting taxonomy of TVGs. 

Then, Section \ref{sub:complexity} focuses in complexity measures in this model. We think that a computational model without correct time complexity measure(s) is not complete. We are unable to find in previous works any such measure that is suitable for all TVGs. In consequence, we propose in this paper a new time complexity measure that captures the cost of the algorithm independently of delays introduced by topology changes and asynchronous communications.

\subsection{Model}\label{sub:model}

Let us first borrow the formalism introduced in~\cite{CFQS12} in order to describe the distributed systems prone to high dynamic. We consider {\em distributed systems} made of $n$ computing entities, henceforth indifferently referred to as {\em nodes}, {\em vertices}, or {\em processes}.  A process has a local memory, a local sequential and deterministic algorithm, and input\slash output capabilities. We assume that each entity has a unique identifier. Moreover, given two distinct entities $p$ and $q$ identified respectively by $id_p$ and $id_q$, either $id_p<id_q$ or $id_q<id_p$. All these entities are gathered in a set $V$.  Let $E$ be a set of edges (or relations) between pairwise entities, that describes interactions between processes, namely communication exchange. The presence of an edge between two vertices $p$ and $q$ at a given time $t$ means that each vertex among $\{p,q\}$ is able to send a message to the other at $t$.

The interactions between processes are assumed to take place over a time span $\mathcal{T} \subseteq \mathbb{T}$ called the {\em lifetime} of the system. The temporal domain $\mathbb{T}$ is generally assumed to be either $\mathbb{N}$ (discrete-time systems) or $\mathbb{R}^+$ (continuous-time systems).

\begin{definition}[Time-varying graph \cite{CFQS12}]
\label{def:TVG}
A time-varying graph (TVG for short) $g$ is a tuple $(V,E,\mathcal{T},\rho,$ $\zeta,\phi)$ where $V$ is a (static) set of vertices $\{v_1,\ldots,v_n\}$, $E$ a (static) set of edges between these vertices $E\subseteq V\times V$, $\rho:E\times\mathcal{T}\to\{0,1\}$ (called presence function) that indicates whether a given edge is available (\emph{i.e.} present) at a given time, $\zeta:E\times\mathcal{T}\rightarrow \mathbb{T}$ (called edge latency function) indicates the time it takes to cross a given edge if starting at a given date, and $\phi:V\times\mathcal{T}\rightarrow \mathbb{T}$ (called process latency function) indicates the time an internal action of a process takes at a given date.
\end{definition}

Given a TVG $g$, let $\mathcal{T}_g$ be the subset of $\mathcal{T}$ for which a topological event (appearance/disappearance of an edge) occurs in $g$. The evolution of $g$ during its lifetime $\mathcal{T}$ can be described as the sequence of graphs $\mathcal{S}_{g} = g_1, g_2,\ldots$, where $g_i=(V,E_i)$ corresponds to the static {\em snapshot} of $g$ at time $t_i \in \mathcal{T}_g$, \emph{i.e.} $e\in E_i$ if and only if $\forall t\in[t_i,t_{i+1}[,\rho(e,t) = 1$. Note that, by definition, $g_i \neq g_{i+1}$ for any $i$. 

We consider {\em asynchronous} distributed systems, \emph{i.e.} no pair of processes has access to any kind of shared device that could allow to synchronize their execution rate.  Furthermore, at any time, no process has access to the output of $\zeta$, \emph{i.e.} none of them can ({\em a priori}) predict a bound on the message delay. Note that the ability to send a message to another process at a given time does not mean that this message will be delivered. Indeed, the dynamicity of the communication graph implies that the edge between the two processes may disappear before the delivery of this message leading to the lost of messages in transit. 

The presences and absences of an edge are instantly detected by its two adjacent nodes. We assume that our system provides to each process a non-blocking communication primitive named \textbf{Send\_retry} that ensures the following property. When a process $p$ invokes \textbf{Send\_retry}$(m,q)$ (where $m$ is an arbitrary message and $q$ another process of $V$) at time $t$, this primitive delivers $m$ to $q$ in a finite time provided that there exists a time $t'\geq t$ such that the edge $\{p,q\}$ is present at time $t'$ during at least $\zeta(\{p,q\},t')$ units of time. In other words, the delivery of the message is ensured if there is, after the invocation of the primitive, an availability of the edge that is sufficient to overcome the communication delay of the edge at this time. Note that this primitive may never deliver a message (\emph{e.g.} if the considered edge never appears after invocation). Details of the implementation of this primitive are not considered here but it typically consists in resending $m$ at each apparition of the edge $\{p,q\}$ until its reception by $q$. This primitive allows us to abstract from topology changes and asynchronous communication and to write high-level algorithms.

\paragraph{Configurations and executions} The \emph{state} of a process is defined by the values of its variables. Given a TGV $g$, a \emph{configuration} of $g$ is a vector of $n+2$ components $(g_i, M_i, p_1, p_2, \ldots, p_n)$ such that $g_i$ is a static snapshot of $g$ (\emph{i.e.} $g_i \in \mathcal{S}_{g}$), $M_i$ is the set of multi-sets of messages carried over $E_i$, and $p_1$ to $p_n$ represent the state of the $n$ processes in $V$. We say that a process $p$ outputs a value $v$ in a configuration $\gamma$ if one of its variable (called an output variable) has the value $v$ in $\gamma$.

An {\em execution} of the distributed system modeled by $g$ is a sequence of configurations $e=\gamma_{0}, \ldots ,\gamma_{k},$ $\gamma_{k+1}, \ldots$, such that for each $k\geq 0$, during an execution step $(\gamma_k,\gamma_{k+1})$, one of the following event occurs: $(i)$ $g_{k} \neq g_{k+1}$, or $(ii)$ at least one process receives a message, sends a message, or executes some internal actions changing its state. The \emph{algorithm} executed by $g$ describes the set of all allowed internal actions of processes (in function of their current state or external events as message receptions or time-out expirations) during an execution of $g$. We assume that during any configuration step $(\gamma_{k},\gamma_{k+1})$ of an execution, if  $g_{k} \neq g_{k+1}$, then for each edge $e$ such that $e \in E_k$  and $e \notin E_{k+1}$ (\emph{i.e.} $e$ disappears during the step $(\gamma_k,\gamma_{k+1}$), none of the messages carried by $e$ belongs to $M_{k+1}$. Also, for each edge $e$ such that $e \in E_{k+1}$  and $e \notin E_{k}$ (\emph{i.e.} $e$ appears during the step $(\gamma_k,\gamma_{k+1})$), $e$ contains no message in configuration $\gamma_{k+1}$. 

\paragraph{Connected over time TVGs} A key concept of time-varying graphs has been identified in~\cite{CFQS12}. The authors shows that the classical notion of path in static graphs in meaningless in TVGs. Indeed, some processes may communicate even if there is no (static) path between them at each time. To perform communication between two processes, the existence of a \emph{temporal path} (\emph{a.k.a.} \emph{journey}) between them is sufficient.  They define such a temporal path as follows: a sequence of ordered pairs $\mathcal{J}=\{(e_1,t_1),(e_2,t_2),...,(e_k,t_k)\}$ such that  $\{e_1,e_2,...,e_k\}$ is a path\footnote{A sequence of edges $\{v_1,v'_1\}, \{v_2,v'_2\}, \ldots, \{v_k,v'_k\}$ is a \emph{path} if  $\forall i \in\{1,k-1\}, v_{i+1} = v'_{i}$.} if for every $i \in [1,k]$, $\rho(e_i,t_i)=1$ and $t_{i+1} \geq t_i + \zeta(e_i,t_i)$. In other words, a journey from process $p$ to process $q$ is a sequence of adjacent edges from $p$ to $q$ such that availability and latency of edges allow the sending of a message from $p$ to $q$ using the \textbf{Send\_retry} primitive at each intermediate process (refer to \cite{CFQS12} for a formal definition). Note that a journey is a non symmetric relation between two processes.

Based on various assumptions made about journeys (\emph{e.g.} recurrence, periodicity, symmetry, and so on), the authors propose in~\cite{CFQS12} proposes a relevant hierarchy of TVG classes. In this paper, we choose to make minimal assumptions on the dynamicity of our system since we restrict ourselves on \emph{connected-over-time} TVGs defined as follows: 

\begin{definition}[Connected-over-time TVG \cite{CFQS12}]\label{def:COT}
A TVG $(V,E,\mathcal{T},\rho,\zeta,\phi)$ is connected-over-time if, for any time $t\in\mathcal{T}$ and for any pair of processes $p$ and $q$ of $V$, there exists a journey from $p$ to $q$ after time $t$. The class of connected-over-time TVGs is denoted by $\mathcal{COT}$\footnote{Authors of \cite{CFQS12} refer to this class as C5 in their hierarchy of TVG classes.}.
\end{definition}

Note that the lifetime of a connected-over-time TVG is necessarily infinite by definition. The class $\mathcal{COT}$
allows us to capture highly dynamic systems since we only require that any process will be always able to communicate
with any other one without any extra assumption on this communication (such as delay, periodicity, or used route). In particular, note that a connected-over-time TVG may be disconnected at each time and that the presence of an edge at a given time does not preclude that this edge will appear again after this time.  Define an \emph{eventual missing edge} as en edge that appears only a finite number of time during the lifetime of the TVG. The main difficulty encountered in the design of distributed algorithms in $\mathcal{COT}$ is to deal with such eventual missing edges because no process is able to predict if a given adjacent edge is an eventual missing edge or not. Note that the time of the last presence of such an eventual missing edge cannot be even bounded.

\begin{definition}[(Eventual) Underlying Graph]\label{def:UG}
Given a TVG $g=(V,E,\mathcal{T},\rho,\zeta,\phi)$, the underlying graph of a $g$ is the (static) graph $U_g=(V,E)$. The eventual underlying graph of $g$ is the (static) subgraph $U^\omega_g=(V,E^\omega_g)$ with $E^\omega_g=E \setminus M_g$, where $M_g$ is the set of eventual missing edges of $g$. 
\end{definition}

Intuitively, the underlying graph (sometimes referred to as {\em footprint}) of a TVG $g$ gathers all edges that appear at least once during the lifetime of $g$, whereas the eventual underlying graph of $g$ gathers all edges that are infinitely often present during the lifetime of $g$. Note that, for any TVG of $\mathcal{COT}$, both underlying graph and eventual underlying graph are connected by definition. Let us define the \emph{neighborhood} $\mathcal{N}_p$ of a process $p$ is the set of processes with which $p$ shares an edge in the underlying graph.

\paragraph{Induced subclasses}

In the following, we focus on specific subclasses of the class $\mathcal{COT}$ to establish our impossibility result. Informally, we focus on subclasses that gather all TVGs whose underlying graph belongs to a given set. The intuition behind this restriction is the following. In practice, some technical reasons may restrict or prevent the communication between some processes, that induces a given underlying graph for the TVG that models our system. In contrast, we cannot predict in general the availabilities and latencies of communication edges, that leads us to consider all TVGs sharing this underlying graph. More formally:

\begin{definition}[Induced subclass]\label{def:inducedsubclass}
Given a set of (static) graphs $\mathcal{F}$ and a class of TVGs $\mathcal{C}$, the subclass of $\mathcal{C}$ induced by $\mathcal{F}$ (denoted by $\mathcal{C}|_\mathcal{F}$) is the set of all TVGs of $\mathcal{C}$ whose underlying graph belongs to $\mathcal{F}$.
\end{definition}

\paragraph{Diameter}

For any given (static) graph $g$, we denote by $diam(g)$ the diameter of $g$ (that is, the longest distance between two processes of $g$).

\subsection{Complexity Measures}\label{sub:complexity}

At the best of our knowledge, there exists currently no time complexity measure that is suitable for any class of TVGs. Some previous works interested in complexity measure in the TVG model but restrict themselves to synchronous systems (see \emph{e.g.} \cite{KOM11c,KLO10c}), to message complexity (see \emph{e.g.} \cite{CFMS10c}), or to specific class of TVGs in which an existing notion of complexity naturally makes sense (see \emph{e.g.} \cite{IKW14c,CFMS10c}).

The first contribution of this paper is to propose a definition of a time complexity measure suitable for our model. We need a definition that captures the ``quality'' of an algorithm independently of delays introduced by asynchronous communications but also by topological changes. A typical example of such a delay is the waiting after the next apparition of an incident edge to a disconnected process that may introduce a long delay that is not imputable to the algorithm but only to the dynamicity of the system. To perform our goal, we propose to extend the classical notion of time complexity commonly adopted in asynchronous message passing (static) systems.

The classical way to deal with communication delays in time complexity measure in asynchronous message passing models is to consider as the unit of time of an execution the worst delay between the sending and the reception of a message during this execution (see \cite{AW04} for example). Using this time measure, we can bound the termination time of any execution of an algorithm independently of communication delays in this execution. This leads to a time complexity measure (the worst termination time over all possible executions of the algorithm) that induces a fair comparison between algorithms. Our proposal is to extend this idea to dynamic environments by including delays introduced by the dynamicity in this definition. In other words, we will consider as the unit of time of an execution the worst delay between the invocation of the \textbf{Send\_retry} primitive and the delivery of the message by this primitive during this execution.

This natural extension of the definition of time complexity measure of asynchronous message passing systems is not sufficient. Indeed, the dynamicity of the system may introduce another possibly arbitrarily long delay that we call initial delay. As an example, consider a problem that requires each process to propagate an initial value (think about consensus-like problems). An easy way to delay the termination of any algorithm for this problem is to disconnect one process for an arbitrary long (but bounded) time that leads all other processes to wait after its first apparition. Intuitively, this delay is not due to the algorithm but to the dynamicity of the system. Consequently, our complexity measure have to ignore such initial delay.

To deal with this issue, we propose to define for each problem a starting time as follows. It is the smallest time of an execution where the dynamicity of the system ``shows'' to processes the minimal topological information to solve the problem. Note that this starting time depends only of the problem (\emph{e.g.} first connexion of the last process for consensus-like problems) and that, in a static system, the starting time and the initial time are identical (since the system cannot delay apparition of any topological information).

Then, we propose to measure the complexity of an algorithm by the worst time (expressed in the time unit described above) between the starting time and the termination of the algorithm over all its possible executions. We believe that this time complexity measure allows us to fairly compare algorithms designed in our model based on TVGs since it exhibits their intrinsic communication costs and does not take in account delays introducing by asynchronous communications and topological changes.

We now state our complexity measure more formally. In the following, we first restrict to \emph{fixed point computation problems} on a TVG class $\mathcal{C}$, \emph{i.e.} problems that admit a specification of the following form: it is required that the execution on every TVG of $\mathcal{C}$ reaches in a finite time a suffix where each process outputs constantly a given value. The required value depends of the considered problem and is not necessarily the same at each process. Using this definition, leader election or spanning structure construction are fixed point computation problems whereas mutual exclusion or broadcast are not.  

We consider now a (deterministic) distributed algorithm $\mathcal{A}$ that satisfies the specification of a fixed point computation problem $\mathcal{P}$ on a TVG class $\mathcal{C}$. Let $e$ be the execution of $\mathcal{A}$ on a given TVG of class $\mathcal{C}$. For any message $m$ sent during $e$, we call \emph{delay} (of $m$) the time between the invocation of the \textbf{Send\_retry} primitive by the sender of $m$ and the delivery of $m$ to its destination. Now, we call \emph{communication step} (or simply step) of $e$ the worst delay over the set of messages that are actually delivered during $e$ (note that we do not consider messages that are never delivered in $e$).

We associate to $\mathcal{P}$ a function $NPS_\mathcal{P}$, called the \emph{necessary presence sets function} of $\mathcal{P}$, that returns, for any TVG $(V,E,\mathcal{T},\rho,\zeta,\phi)$ of $\mathcal{C}$, a set of subsets of $E$. Note that the actual definition of this function depends of the problem itself and not of a TVG nor an execution. Each element of $NPS_\mathcal{P}(g)$ describes one of the set of edges whose apparition is necessary and sufficient to start the effective solving the problem (independently of the used algorithm). We give some examples in the following. For the underlying graph computation problem $\mathcal{UG}$, we have $NPS_\mathcal{UG}(g)=\{E\}$ since each edge of $E$ must appear in the output of any process. For a broadcast problem $\mathcal{B}$, we have $NPS_\mathcal{B}(g)=\{\{(p,q)\}|q\in \mathcal{N}_p\}$ (where process $p$ is the sender of the message) since the apparition of any edge adjacent to $p$ is necessary and sufficient to begin the broadcast of a message by $p$.

We define the \emph{starting time} of the execution $e$ of $\mathcal{A}$ over a TVG $g$ as the smallest time $t\in\mathcal{T}$ such that each edges of at least one element of $NPS_{\mathcal{P}}(g)$ are present at least once before $t$ in this execution. Note that, in a static distributed system, the initial time and the starting time are always identical since all edges of all elements of $NPS_\mathcal{P}(g)$ are present in the initial configuration whatever the definition of $NPS_\mathcal{P}$ is. Finally, the \emph{convergence time} of $\mathcal{A}$ on $g$ is the time (expressed in communication steps of $e$) between the starting time of $e$ and the smallest time in $e$ where the specification of $\mathcal{P}$ is satisfied.

\begin{definition}[Time complexity on a TVG class]
The time complexity of a distributed algorithm $\mathcal{A}$ that satisfies the specification of a fixed point computation problem $\mathcal{P}$ on a TVG class $\mathcal{C}$ is the worst convergence time of $\mathcal{A}$ on all TVGs of $\mathcal{C}$.
\end{definition}

Note that this definition may be naturally extended to so-called \emph{service problems} in the following way. First, we consider as starting time the maximum between the starting time defined above and the time of request of a service (\emph{e.g.} the sending of a message for a broadcast algorithm, the request of critical section for a mutual exclusion algorithm). Second, we substitute the convergence time of the algorithm by the time of achievement of the required service by the algorithm (\emph{e.g.} the delivery of a message to its destinations for a broadcast algorithm, the starting of critical section for a mutual exclusion algorithm).

\section{Underlying Graph Computation}\label{sec:UG}

In this section, we present an underlying graph computation algorithm (see Section \ref{sub:UGalgo}) and proves its time optimality with respect to our new measure (see Section \ref{sub:UGcomplexity}). This algorithm is used as a building block in the next section for our minimal dominating set construction algorithm. Before presenting our algorithm, we need to specify the underlying graph computation problem.

\begin{specification}[Underlying graph]
An algorithm $\mathcal{A}$ satisfies the underlying graph specification for a class of TVGs $\mathcal{C}$ if the execution $e=\gamma_0,\gamma_1,\ldots$ of $\mathcal{A}$ on every TVG $g$ of $\mathcal{C}$ has a suffix $e_i=\gamma_i,\gamma_{i+1},\ldots$ for a given $i\in\mathbb{N}$ such that each process outputs the underlying graph of $g$ in any configuration of $e_i$.
\end{specification}

\subsection{Algorithm}\label{sub:UGalgo}

Our underlying graph computation algorithm is presented in Algorithm \ref{algo:ug}. The intuition behind this algorithm is simple. Each process stores locally a graph, initially empty, that eventually gathers all edges of the underlying graph. At the first appearance of an edge, the two adjacent processes add this edge to their graph. Then, they try to propagate the last version of their graph to all processes that they have as neighbor at least once since the beginning of the execution. When a process receives such a message (that contains the current underlying graph of another process), it add to its own underlying graph every edge it does not already know. If its underlying graph grows during this operation, then the process propagates again its underlying graph to all processes that it has as neighbor at least once since the beginning of the execution.

\begin{algorithm}\caption{Underlying graph computation for process $p$.}\label{algo:ug}
\begin{description}
\item[Variables:]~\\
$g_p=(V_p,E_p)$: underlying graph built by $p$\\
$\mathcal{N}_p$: neighborhood of $p$
\item[Initialization:]~\\
$g_p:=(\{x\},\emptyset)$\\
$\mathcal{N}_p:=\emptyset$
\item[Upon appearance of an edge $\{p,q\}$:]~
\begin{tabbing}
xxx \= xxx \= \kill 
\textbf{if} $\{p,q\}\notin E_p$ \textbf{then}\\
\> $\mathcal{N}_p:=\mathcal{N}_p\cup\{q\}$\\
\> $g_p:=(V_p\cup\{q\},E_p\cup\{\{p,q\}\})$\\
\> \textbf{foreach} $r\in\mathcal{N}_p$ \textbf{do}\\
\>\> \textbf{Send\_retry}$(add(g_p),r)$
\end{tabbing}
\item[On reception of $add(g_q)$ from $q$:]~
\begin{tabbing}
xxx \= xxx \= \kill 
\textbf{if} $E_q\setminus E_p\neq\emptyset$ \textbf{then}\\
\> $g_p:=(V_p\cup V_q,E_p\cup E_q)$\\
\> \textbf{foreach} $r\in\mathcal{N}_p\setminus \{q\}$ \textbf{do}\\
\>\> \textbf{Send\_retry}$(add(g_p),r)$
\end{tabbing}
\end{description}
\end{algorithm}

This algorithm ensures that, upon the first apparition of the last edge of the underlying graph, this edge is added to the output of adjacent processes and then propagated (at least) to their neighbors in the eventual underlying graph in one step, and so on (note that we have no guarantees for neighbors in the underlying graph in general since it may exist some eventual missing edges). Hence, in any execution, after at most $diam(U^\omega_g)$ steps, this edge (and all others) appears in the output graph of any process. In other words, we have the following result:

\begin{theorem}
Algorithm \ref{algo:ug} satisfies the underlying graph specification for $\mathcal{COT}$. Moreover, its convergence time on any TVG $g$ of $\mathcal{COT}$ is $diam(U^\omega_g)$ steps.
\end{theorem}

\subsection{Time Optimality}\label{sub:UGcomplexity}

In this section, we interest in a lower bound result on the time complexity of underlying graph computation. We restrict ourselves to greedy algorithms that are the most natural ones for this problem. We define a {\em greedy algorithm} for the underlying graph computation as an algorithm that satisfies the following property. The initial output of any process is an empty graph and the graph outputted by a process can only grow (in the sense of inclusion) over time. In other words, such an algorithm ensures that, once a process start to output a given edge or process, this latter always appears in the output of this process afterwards. Note that Algorithm~\ref{algo:ug} falls in this category.

In the following, we prove that no greedy algorithm for underlying graph computation on $\mathcal{COT}$ can exhibit a better time complexity than our algorithm. Indeed, we prove that there exists, for any greedy algorithm, a TVG $g$ in $\mathcal{COT}$ such that this algorithm needs $diam(U^\omega_g)$ steps to compute the underlying graph of $g$. Note that the complexity of the underlying graph computation depends surprisingly of a parameter of the \emph{eventual} underlying graph. 

We need a technical lemma for the proof of this optimality result.

\begin{lemma}\label{lem:optimality}
For any greedy algorithm $\mathcal{A}$ that satisfies the underlying computation graph, for any TVG $g=(V,E,\mathcal{T},\rho,\zeta,\phi)$ in $\mathcal{COT}$, for any edge $e\in E$ that is not a cut-edge of $U^\omega_g$, for any process $p\in V$, for any $t\in\mathcal{T}$, $e$ cannot belong to the graph outputted by $p$ in the execution of $\mathcal{A}$ on $g$ at time $t$ if there exists no temporal path from one extremity of $e$ to $p$ that starts after the first appearance of $e$ in $g$ and ends before $t$.
\end{lemma}

\begin{proof}
By contradiction, assume that there exists a greedy algorithm $\mathcal{A}$ that satisfies the underlying computation graph, a TVG $g=(V,E,\mathcal{T},\rho,\zeta,\phi)$ in $\mathcal{COT}$, an edge $\hat{e}\in E$ that is not a cut-edge of $U^\omega_g$, a process $\hat{p}\in V$, and a time $\hat{t}\in\mathcal{T}$ such that $\hat{e}$ appears in the graph outputted by $\hat{p}$ in the execution of $\mathcal{A}$ on $g$ at time $\hat{t}$ and that there exists no temporal path from one extremity of $\hat{e}$ to $\hat{p}$ that starts after the first appearance of $\hat{e}$ in $g$ and ends before $\hat{t}$.

Then, consider the TVG $g'=(V,E,\mathcal{T},\rho',\zeta,\phi)$ with:
\[
\forall e\in E,\forall t\in\mathcal{T},\rho'(e,t)=\begin{cases}
0 \text{ if } e=\hat{e}\\
\rho(e,t) \text{ otherwise}
\end{cases}
\]

Note that, according to the assumption that $\hat{e}$ is not a cut-edge of $U^\omega_g$, $g'$ belongs to $\mathcal{COT}$. Hence, due to the construction of $g$ and the determinism of $\mathcal{A}$, the process $\hat{p}$ receives exactly the same messages before time $\hat{t}$ in $g$ and $g'$ (the assumption on temporal paths after the first appearance of $\hat{e}$ in $g$ ensures us that the fact to remove $\hat{e}$ in $g'$ is not detectable by $\hat{p}$ before time $\hat{t}$). In other words, $\hat{p}$ cannot distinguish the executions of $\mathcal{A}$ on $g$ and $g'$ before $\hat{t}$. As a consequence, $\hat{e}$ appears in the graph outputted by $\hat{p}$ in the execution of $\mathcal{A}$ on $g'$ at time $\hat{t}$. As $\mathcal{A}$ is a greedy algorithm, this edge never disappear of the output of $\hat{p}$ in the execution of $\mathcal{A}$ on $g'$ after time $\hat{t}$. This is contradictory with the fact that $\mathcal{A}$ satisfies the underlying graph specification on $\mathcal{COT}$ since $\hat{e}$ does not belongs to $U_{g'}$ and proves the lemma.
\end{proof}

We are now ready to prove the following result.

\begin{theorem}\label{th:optimality}
For any greedy algorithm $\mathcal{A}$ that satisfies the underlying graph specification on $\mathcal{COT}$, there exists a TVG $g$ of $\mathcal{COT}$ such that the convergence time of $\mathcal{A}$ is at least $diam(U^\omega_g)$ steps.
\end{theorem}

\begin{proof}
Let $\mathcal{A}$ be a greedy algorithm that satisfies the underlying graph specification on $\mathcal{COT}$. Then, let us define the following family of TVGs. For any given $k\in\mathbb{N}^*$, let $g_k=(V_k,E_k,\mathbb{R}^+,\rho_k,\zeta_k,\phi_k)$ be the TVG defined by $V_k=\{p_0,\ldots,p_{3k}\}$, $E_k=\{\{p_i,p_{i+1}\}|i\in\{0,\ldots,3k-1\}\}\cup\{\{p_0,p_{2k}\},\{p_{2k},p_{3k}\}\}$, and
\[\forall e\in E_k,\forall t\in\mathbb{R}^+,\rho_k(e,t)=
\begin{cases}
1 \text{ if } e\in\{\{p_0,p_{2k}\},\{p_{2k},p_{3k}\}\} \text{ and } t< 1\\
0 \text{ if } e\notin\{\{p_0,p_{2k}\},\{p_{2k},p_{3k}\}\} \text{ and } t< 1\\
0 \text{ if } e\in\{\{p_0,p_{2k}\},\{p_{2k},p_{3k}\}\} \text{ and } t\geq 1\\
1 \text{ if } e\notin\{\{p_0,p_{2k}\},\{p_{2k},p_{3k}\}\} \text{ and } t\geq 1
\end{cases}
\]
\[\forall e\in E_k,\forall t\in\mathbb{R}^+,\zeta_k(e,t)=1\]
\[\forall p\in V_k,\forall t\in\mathbb{R}^+,\phi_k(p,t)=0\]

\begin{figure}
  \centering 
  \includegraphics[scale=0.4]{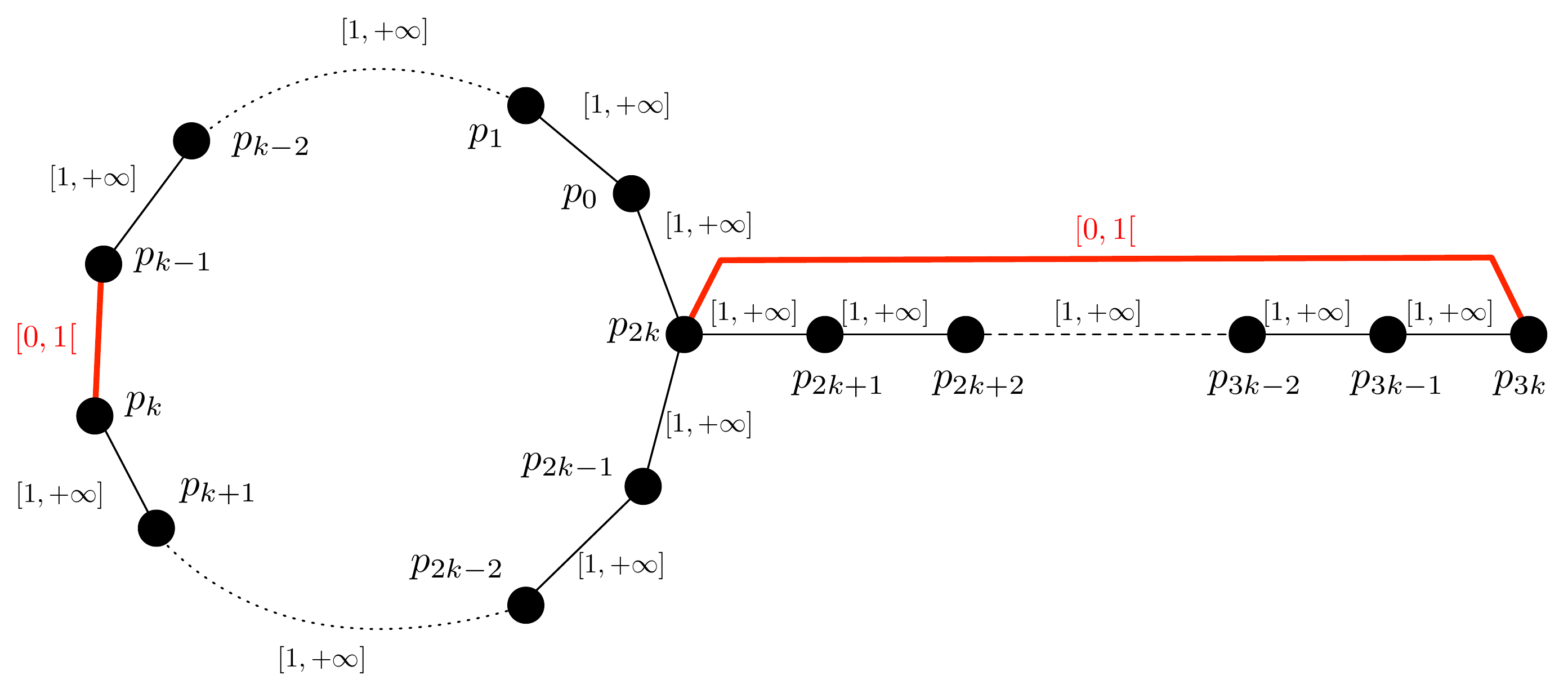}
  \caption{An illustration of the TVGs family in the proof of Theorem \ref{th:optimality}. \label{fig:optimality}}
\end{figure}

Refer to Figure \ref{fig:optimality} for an example of such a $g_k$. Note that, for any $k\in\mathbb{N}^*$, we have $diam(U^\omega_{g_k})=2k$ (and $diam(U_{g_k})<diam(U^\omega_{g_k})$ since $diam(U_{g_k})=k+1k$). As this graph is connected, $g_k$ belongs to $\mathcal{COT}$. By construction of $g_k$, the starting time of the execution of $\mathcal{A}$ on $g_k$ is $1$ for any $k\in\mathbb{N}^*$ (recall that $NPS_\mathcal{UG}(g)=\{E\}$). Note that, due to the choice of the latency function, any communication step of the execution of $\mathcal{A}$ on $g_k$ takes exactly one time unit.

Consider $e_k$ the execution of $\mathcal{A}$ on $g_k$ for any $k\in\mathbb{N}^*$. From Lemma \ref{lem:optimality}, we know that the edge $\{p_{k-1},p_k\}$ cannot appear in the graph outputted by $p_{3k}$ in $e_k$ before there exists at least one temporal path from $p_{k-1}$ or $p_k$ to $p_{3k}$. Note that the construction of $g_k$ implies that such a temporal path (after time $1$) needs at least $2k$ steps (the length of the path from $p_{k-1}$ or $p_k$ to $p_{3k}$ since $g_k$ is static after time $1$). As the edge $\{p_{k-1},p_k\}$ must eventually appear in the output of any process in $e_k$ by assumption on $\mathcal{A}$, we obtain that the convergence time of $\mathcal{A}$ is at least $diam(U^\omega_{g_k})$ steps, that ends the proof.
\end{proof}

\section{Minimal Dominating Set Construction}\label{sec:MDS}

Minimal dominating set construction is a classical problem in the context of distributed computing since this spanning structure have interesting properties for a lot of practical problems as clustering. Recall that, in a static distributed system, a dominating set $D$ is a subset of processes of the system such that each process that does not belong to $D$ have at least one neighbor in $D$. Such a dominating set is minimal when it has is no strict subset that is also a dominating set.

Regarding dynamic distributed systems, two different approaches have been proposed to handle minimal dominating set problem. We survey them quickly here and show that these definitions seem not relevant in our context, that motivates the need of our new definition presented in this section.

The most natural way to extend minimal dominating set definition in the context of dynamic systems is presented in \cite{WDCG12}. In this work, the dynamic graph is seen as a sequence of static graphs and a new minimal dominating set is computed at each topological change. This approach is not suitable in the case of highly dynamic systems since the system may be always in computation phase (the computation of the new dominating set at each topological change is not instantaneous). In this case, the dominating set may be never stable and is then useless for the application that required it.

The second approach, proposed by \cite{CF13r}, consists in computing a stable dominating set on the underlying graph of the TVG. This approach is interesting since the outputted dominating set is stable in spite of the dynamicity of the system but is still not suitable for our purpose. Indeed, as the dominating set is computed on the underlying graph that may contain eventual missing edges, it is possible for a process to be dominated only through such edges. In other words, a dominated process may have eventually only dominated neighbors, that is counter-intuitive for a minimal dominating set and makes sense only in TVGs where there is no eventual missing edges.

To overcome flaws of precedent definitions in our context of highly dynamic distributed systems (captured by the class of TVGs $\mathcal{COT}$), we propose a third definition. In this definition,  we require the outputted minimal dominating set to be stable and each dominated process to be infinitely often neighbor of at least one dominating process. In other words, we want to compute a minimal dominating set on the \emph{eventual} underlying graph. Note that this definition is exactly the same as the one of \cite{CF13r} in TVGs where there is no eventual missing edges. 

\begin{definition}[Minimal dominating set over time] 
A set of processes $M$ is a minimal dominating set over time (MDST for short) of a TVG $g$ if $M$ is a minimal dominating set of $U^\omega_g$.
\end{definition}

We now specify the minimal dominating set construction problem over TVGs as follows.

\begin{specification}[Minimal dominating set]
An algorithm $\mathcal{A}$ satisfies the minimal dominating set specification for a class of TVGs $\mathcal{C}$ if the execution $e=\gamma_0,\gamma_1,\ldots$ of $\mathcal{A}$ on every TVG $g$ of $\mathcal{C}$ has a suffix $e_i=\gamma_i,\gamma_{i+1},\ldots$ for a given $i\in\mathbb{N}$ such that each process outputs constantly a boolean value in any configuration of $e_i$ and that the set of processes outputting true is a minimal dominating set overt time of $g$.
\end{specification}

\subsection{Preliminaries}\label{sub:MDSprelem}

In this section, we present some preliminary results that are needed in the following. First, we introduce the definition of a strong minimal dominating set of a graph as a dominated set of any connected spanning subgraph of this graph. In Section \ref{sub:MDSimp}, we prove that the existence of such a set in the underlying graph of a TVG is necessary to the existence of an algorithm to construct a minimal dominating set over time of this TVG. We claim in Section \ref{sub:MDSalgo} that this condition is also sufficient. To prove this result, we use a characterization of graphs that admit a strong minimal dominating set that we present in the end of this preliminary section.
 
\begin{definition}[Strong minimal dominating set] 
A strong minimal dominating set (SMDS for short) of a (static) graph $g$ is a subset of processes of $g$ that is a minimal dominating set of every connected spanning subgraph of $g$.
\end{definition}

The following lemma follows directly from definitions and legitimates our interest for strong minimal dominating sets.

\begin{lemma}\label{lem:smdsismdst}
If the underlying graph of a TVG $g\in \mathcal{COT}$ admits a strong minimal dominating set $M$ then $M$ is a minimal dominating set over time of $g$.
\end{lemma}

The next result provides us a characterization of (static) graphs that admits a SMDS. We use this characterization in our minimal dominating set construction algorithm in the next section.

\begin{lemma}\label{lem:characterization}
For any (static) graph $g$ and any minimal dominating set $M$ of $g$, $M$ is a strong minimal dominating set of $g$ if and only if the set of edges $\{\{p,q\}| q\in M\cap\mathcal{N}_p\}$ is a cut-set in $g$ for every process $p\in V\setminus M$.
\end{lemma}

\begin{proof}
First, we prove that, for any SMDS $M$ of a graph $g$, the set of edges $\{\{p,q\}| q\in M\cap\mathcal{N}_p\}$ is a cut-set in $g$ for every process $p\in V\setminus M$. By contradiction, assume that there exists a SMDS $M$ of a graph $g=(V,E)$ such that the set of edges $\{\{p,q\}| q\in M\cap\mathcal{N}_p\}$ is not a cut-set in $g$ for a process $p\in V\setminus M$. Let $s_g=(V,E')$ the subgraph of $g$ defined by $E'=E\setminus\{\{p,q\}| q\in M\cap\mathcal{N}_p\}$. By assumption, $s_g$ is a connected graph. Moreover, in $s_g$ the process $p$ has no neighbor in $M$, that means that $M$ is not a minimal dominating set of $s_g$. This contradicts the fact that $M$ is a SMDS of $g$ and proves the necessity of the condition.

Second, we prove that any minimal dominating set $M$ of a graph $g$ such that the set of edges $\{\{p,q\}| q\in M\cap\mathcal{N}_p\}$ is a cut-set in $g$ for every process $p\in V\setminus M$ is a SMDS of $g$. By contradiction again, assume that there exists a minimal dominating set $M$ of a graph $g=(V,E)$ such that the set of edges $\{\{p,q\}| q\in M\cap\mathcal{N}_p\}$ is a cut-set in $g$ for every process $p\in V\setminus M$ is not a SMDS of $g$. By definition of a SMDS, there exists a connected subgraph $s_g=(V,E')$ of $g$ such that $M$ is not a minimal dominating set of $s_g$. Let us study the two following cases.

\begin{enumerate}
\item $M$ is not a dominating set of $s_g$. Then, there exists a process $p$ such that no neighbors of $p$ in $s_g$ belongs to $M$. As $s_g$ is connected, that means that the set $\{\{p,q\}| q\in M\cap\mathcal{N}_p\}$ is not a cut-set in $g$, that is contradictory with the initial assumption on $M$.
\item $M$ is a dominating set of $s_g$ but is not minimal. We say that a process of a dominating set properly dominates one of its neighbor if it is the only dominating process in the neighborhood of this dominated process. Then, we know that there exists, in $s_g$, two neighbors $p\in M$ and $q\in M$ such that $p$ does not dominate properly any of its neighbors. As $M$ is a minimal dominating set in $g$, we deduce that $p$ dominates properly at least one of its neighbors $r\in V\setminus M$ (recall that $p$ and $q$ are neighbors in $g$ by construction). That means that the set of edges $\{\{r,s\}| s\in M\cap\mathcal{N}_r\}=\{\{p,r\}\}$ is not a cut-set of $g$ (since this edge does not belong to $s_g$ that is connected). This is contradictory with the initial assumption on $M$.
\end{enumerate}

These contradictions show us the sufficiency of the condition and ends the proof.
\end{proof}

\subsection{Impossibility Result}\label{sub:MDSimp}

\begin{figure*}[htb]
  \centering 
  \includegraphics[scale=0.4]{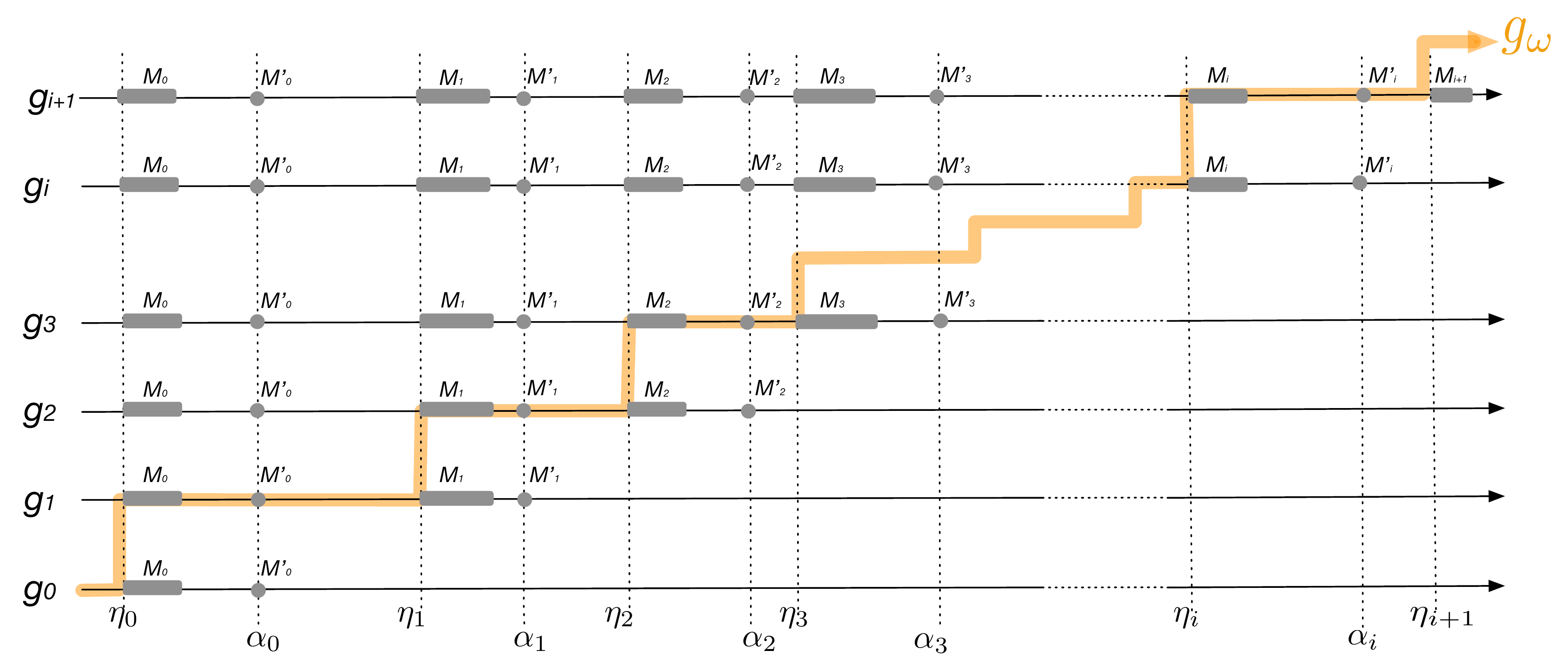}
  \caption{An illustration of the sequence $(g_n)_{n\in\mathbb{N}}$ used in the proof of Theorem \ref{th:MDSimp}. \label{fig:impossibility}}
\end{figure*}

The proof of our impossibility result presented in Theorem \ref{th:MDSimp} makes use of a generic framework we proposed in another work. We recall here the minimal definitions and results to understand our proof. Due to the lack of space, the interested reader is referred to \cite{BDKP14r} for more details. 

\paragraph{Summary of \cite{BDKP14r}} For a given time domain $\mathbb{T}$, a given static graph $(V,E)$ and a given latency function $\zeta$, let us consider the set $\mathcal{G}_{(V,E),\mathbb{T},\zeta}$ of all TVGs over $\mathbb{T}$ that admit $(V,E)$ as underlying graph and $\zeta$ as latency function. For the sake of clarity, we will omit the subscript $(V,E),\mathbb{T},\zeta$ and simply denote this set by $\mathcal{G}$. Remark that two distinct TVGs of $\mathcal{G}$ can be distinguished only by their presence function. For any TVG $g$ in $\mathcal{G}$, let us denote its presence function by $\rho_g$. We define now the following metric $d_\mathcal{G}$ over $\mathcal{G}$. If $g=g'$, then $d_\mathcal{G}(g,g')=0$. Otherwise, $d_\mathcal{G}(g,g')=2^{-\lambda}$ with $\lambda = \text{Sup }\{ t\in\mathbb{T}|\forall t'\leq t,\forall e\in E,\rho_g(e,t') = \rho_{g'}(e,t')\}$.

For a given algorithm $\mathcal{A}$ and a given TVG $g$, let us define the $(\mathcal{A},g)$-output as the function that associate to any time $t\in\mathbb{T}$ the state of $g$ at time $t$ when it executes $\mathcal{A}$. We say that $g$ is the supporting TVG of this output. Let us consider the set $\mathcal{O}_{\mathcal{A},\mathcal{G}}$ of all $(\mathcal{A},g)$-outputs over all TVGs $g$ of $\mathcal{G}$. For the sake of clarity, we will omit the subscript $\mathcal{A},\mathcal{G}$ and simply denote this set by $\mathcal{O}$. Remark that two distinct output of $\mathcal{O}$ can be distinguished only by their supporting TVG. For any output $o$ in $\mathcal{O}$, let us denote its supporting TVG by $g_o$. We define now the following metric $d_\mathcal{O}$ over $\mathcal{O}$. If $o=o'$, then $d_\mathcal{O}(o,o')=0$. Otherwise, $d_\mathcal{O}(o,o')=2^{-\lambda}$ with $\lambda = \text{Sup }\{ t\in\mathbb{T}|\forall t'\leq t,o(t') = o'(t')\}$.

Once we have observed that the metric spaces $(\mathcal{G},d_\mathcal{G})$ and $(\mathcal{O},d_\mathcal{O})$ are complete, we are now able to recall the main result of \cite{BDKP14r}. Intuitively, this theorem ensures us that, if we take a sequence of TVGs with ever-growing common prefixes, then the sequence of corresponding outputs also converges. Moreover, we are able to describe the output to which it converges as the output that corresponds to the TVG that shares all commons prefixes of our TVGs sequence. This result is useful since it allows us to construct counter-example in the context of impossibility results. Indeed, it is sufficient to construct a TVG sequence (with ever-growing common prefixes) and to prove that their corresponding outputs violates the specification of the problem for ever-growing time to exhibit an execution that violates infinitely often the specification of the problem. More formally, we have:

\begin{theorem}\label{th:convergence}
For any deterministic algorithm $\mathcal{A}$, if a sequence $(g_n)_{n\in\mathbb{N}}$ of $\mathcal{G}$ converges to a given $g_\omega\in\mathcal{G}$, then the sequence $(o_n)_{n\in\mathbb{N}}$ of the $(\mathcal{A},g_n)$-outputs converges to $o_\omega\in\mathcal{O}$. Moreover, $o_\omega$ is the $(\mathcal{A},g_\omega)$-output.
\end{theorem}

\paragraph{Application to minimal dominating set}

We are now in measure to prove our impossibility result. This result states that there exists no deterministic algorithm that satisfies the minimal dominating set specification on a TVG of $\mathcal{COT}$ as soon as the underlying graph of the considered TVG does not admit a strong minimal dominating set. Intuitively, this impossibility comes from the following fact. As no process is able to detect eventual missing edges, the minimal dominated set computed by any algorithm must be a minimal dominated set of any possible eventual underlying graph, that is of any connected subgraph of the underlying graph. In other words, the computed minimal dominated set is a strong minimal dominating set. The existence of such a set is then a necessary condition to the existence of an algorithm to compute a minimal dominating set over time. The main difficulty of the formal proof of this result lies in the construction of the TVGs sequence that allows us to apply Theorem \ref{th:convergence}.

\begin{theorem}\label{th:MDSimp}
For any set of (static) graphs $\mathcal{F}$ containing at least one graph that does not admit a strong minimal dominating set, there exists no deterministic algorithm that satisfies the minimal dominating set specification for $\mathcal{COT}|_\mathcal{F}$.
\end{theorem}

\begin{proof}
Let us introduce some notation first. We define, for any TVG $g=(V,E,\mathcal{T},\rho,\zeta,\phi)$,  the TVG $g\odot\{(E_i,\mathcal{T}_i)|i\in I\}$ (with $I\subseteq \mathbb{N}$ and for any $i\in I$, $E_i\subseteq E$ and $\mathcal{T}_i\subseteq\mathcal{T}$) as the TVG $(V,E,\mathcal{T},\rho',\zeta,\phi)$ with:
\[\rho'(e,t)=\begin{cases}
0 \text{ if } \exists i\in I, e\in E_i \text{ and } t\in\mathcal{T}_i\\
1 \text{ if } \exists i\in I, e\in E\setminus E_i \text{ and } t\in\mathcal{T}_i\\
\rho(e,t) \text{ otherwise}
\end{cases}\]

By contradiction, assume that there exists a set of (static) graphs $\mathcal{F}$ containing at least one graph that does not admit a strong minimal dominating set and that there exists a deterministic algorithm $\mathcal{A}$ that satisfies the minimal dominating set specification for $\mathcal{COT}|_\mathcal{F}$. In consequence, any process that executes $\mathcal{A}$ outputs a boolean value at any time.

Let $g=(V,E,\mathcal{T},\rho,\zeta,\phi)$ be a TVG of $\mathcal{COT}|_\mathcal{F}$ such that $U_g$ does not admit a strong minimal dominating set and that all edges of $U_g$ are present during the first communication step of the execution of $\mathcal{A}$ on $g$ ($g$ exists by construction of $\mathcal{F}$ and by definition of $\mathcal{COT}|_\mathcal{F}$). Let $t_0$ be the time of completion of the first communication step of the execution of $\mathcal{A}$ on $g$. We construct then a sequence $(g_n)_{n\in\mathbb{N}}$ of TVGs as follows. We set $g_0=g$. Assume that we have already $g_i=(V,E,\mathcal{T},\rho',\zeta,\phi)$ for a given $i\in\mathbb{N}$ such that $g_i\in\mathcal{COT}|_\mathcal{F}$, $U_{g_i}=U_g$, and $\exists \alpha_i>t_0,\forall e\in E,\forall t\leq \alpha_i,\rho'(e,t)=\rho(e,t)$. Then, we define inductively $g_{i+1}$ as follows (refer to Figure \ref{fig:impossibility} for an illustration, gray boxes represent portions of executions where $\mathcal{A}$ outputs a stable minimal dominating set):

\begin{enumerate}
\item Consider the execution of $\mathcal{A}$ over $g_i$ and let $\eta_{i}\in\mathcal{T}$ be the smallest time strictly greater than $\alpha_i$ from which the set of processes that output true is constant ($\eta_i$ exists by assumption on $\mathcal{A}$ since $g_i\in\mathcal{COT}|_\mathcal{F}$);
\item Let $M_i$ be the minimal dominating set computed by $\mathcal{A}$ on $g_i$ (\emph{i.e.} the set of processes of $g_i$ outputting true after $\eta_i$). As $U_{g_i}=U_g$, we know by assumption on $U_g$ that $U_{g_i}$ does not admit a SMDS. In particular, $M_i$ is not a SMDS of $U_{g_i}$. Hence, there exists a process $p_i$ of $V\setminus M_i$ such that the set of edges  $E_i=\{\{p_i,q\}| q\in M_i\cap\mathcal{N}_{p_i}\}$ is not a cut-set of $U_{g_i}$; 
\item Let $g'_i=g_i\odot\{(E_i,\mathcal{T}\cap]\eta_i,+\infty[)\}$.
\item Remark that $U_{g'_i}=U_{g_i}=U_g$ (by construction of $g'_i$ since $\eta_i>t_0$) and that $U^\omega_{g'_i}$ is connected (since $E(U^\omega_{g'_i})=E(U_g)\setminus E_i$ by construction\footnote{where $E(g)$ denotes the set of edges of $g$.} and $E_i$ is not a cut-set of $U_{g}$). Hence, $g'_i\in\mathcal{COT}|_\mathcal{F}$ and we can consider the execution of $\mathcal{A}$ over $g'_i$. Let $\alpha_{i}\in\mathcal{T}$ be the smallest time strictly greater than $\eta_i$ from which the set of processes that output true is constant. Let $M'_i$ be the minimal dominating set computed by $\mathcal{A}$ on $g'_i$ (\emph{i.e.} the set of processes of $g'_i$ outputting true after $\alpha_i$). Note that $M'_i\neq M_i$ since $M_i$ is not a minimal dominating set of $U^\omega_{g'_i}$ (recall that, in $U^\omega_{g'_i}$, $p_i$ has no neighbor in $M_i$);
\item Let $g_{i+1}=g_i\odot\{(E_i,\mathcal{T}\cap]\eta_i,\alpha_i])\}$.
\end{enumerate}

It is straightforward to check that this construction ensures that, if there exists $g_i=(V,E,\mathcal{T},\rho',\zeta,\phi)$ for a given $i\in\mathbb{N}$ such that $g_i\in\mathcal{COT}|_\mathcal{F}$, $U_{g_i}=U_g$, and $\exists \alpha_i>t_0,\forall e\in E,\forall t\leq \alpha_i,\rho'(e,t)=\rho(e,t)$, then $g_{i+1}$ satisfies the same property. Moreover, as $g_0=g$, this property is naturally satisfied for $i=0$ with any $\alpha_0>t_0$. Hence, the sequence $(g_n)_{n\in\mathbb{N}}$ is well-defined. Note that, for any $i\in\mathbb{N}$, $\eta_i<\alpha_i$ and $\alpha_i<\eta_{i+1}$ (by construction). 

That allows us to define the following TVG: $g_\omega=g\odot\{(E_i,\mathcal{T}\cap]\eta_i,\alpha_i])|i\in\mathbb{N}\}$. Note that $U_{g_\omega}=U_g$ and then that $g_\omega$ belongs to $\mathcal{COT}|_\mathcal{F}$. Observe that, for any $k\in\mathbb{N}^*$, we have $d_\mathcal{G}(g_k,g_\omega)=2^{-\eta_k}$ by construction of $(g_n)_{n\in\mathbb{N}}$ and $g_\omega$. Thus, $(g_n)_{n\in\mathbb{N}}$ converges in $\mathcal{COT}|_\mathcal{F}$ to $g_\omega$.

We are now in measure to apply the Theorem \ref{th:convergence} that states that the $(\mathcal{A},g_\omega)$-output is the limit of the sequence of the $(\mathcal{A},g_n)$-outputs. In other words, the $(\mathcal{A},g_\omega)$-output shares a prefix of length $\eta_i$ with the $(\mathcal{A},g_i)$-output for any $i\in\mathbb{N}$ (recall that the sequence of the $(\mathcal{A},g_n)$-outputs is Cauchy since it converges). That means that, for any $i\in\mathbb{N}^*$, the set of processes that output true in $g_\omega$ at $\eta_i$ is $M_i$ and the set of processes that output true in $g_\omega$ at $\alpha_i$ is $M'_i$. As we know that $M_i\neq M'_i$ for any $i\in\mathbb{N}$, we obtain that the set of processes that output true in $g_\omega$ never converges, that contradicts the fact that $\mathcal{A}$ satisfies the minimal dominating set specification for $\mathcal{COT}|_\mathcal{F}$ and ends the proof.
\end{proof}

\subsection{Algorithm}\label{sub:MDSalgo}

We are now able to prove the sufficiency of the existence of a strong minimal dominating set on the underlying graph for the construction of a minimal dominating set over time of any TVG of $\mathcal{COT}$. We prove this result simply by presenting an algorithm based on our underlying graph computation algorithm presented in Section \ref{sec:UG}. 

This algorithm works as follows. Once a process has computed the underlying graph, it is easy to decide if this process belongs to the outputted minimal dominating set: the process enumerates (locally and in a deterministic order based \emph{e.g.} on process identities) all minimal dominating sets of the underlying graph and chooses the first one that satisfies Lemma \ref{lem:characterization}. This latter is then a strong minimal dominating set of the underlying graph and hence a minimal dominating set over time of the TVG by Lemma \ref{lem:smdsismdst}. In order to avoid the use of an algorithm of termination detection (for the underlying graph computation), each process repeats the local computation of its output at each update of its local copy of the underlying graph by the algorithm of Section \ref{sec:UG}. The existence of this simple algorithm allows us to state the following result:

\begin{theorem}
For any set of (static) graphs $\mathcal{F}$ containing only graphs that admit a strong minimal dominating set, there exists a deterministic algorithm that satisfies the minimal dominating set specification for $\mathcal{COT}|_\mathcal{F}$.
\end{theorem}

\section{Conclusion}\label{sec:conclu}

This paper addressed the construction of a minimal dominating set over time (MDST) in highly dynamic distributed systems. We considered the weakest connectivity assumption in the hierarchy of time-varying graphs: the graph may be disconnected at each time, topological changes are unpredictable but we know that any process is able to communicate with any other infinitely often using so-called temporal paths. In this context, we proposed a new definition of minimal dominating set increasing the stability of the computed MDST. Next, we provided a necessary and sufficient topological condition for the existence of a deterministic MDST algorithm. We then proposed a new measure of time complexity that takes in account the communication delays due to network dynamic.  

The above results used the construction of an underlying graph. We showed the time optimality of our algorithm with respect to our measure. Note that our result (Theorem~\ref{th:optimality}) is valid for greedy algorithms only. We conjecture that all distributed underlying graph algorithms are greedy. This would lead to generalize our result of optimality. Also, we would like to extend our approach to other related overlay constructions. 

\bibliographystyle{plain}
\bibliography{biblioRR}

\end{document}